\documentclass[12pt,draftclsnofoot,onecolumn]{IEEEtran}

\usepackage{setspace}
\usepackage{amsfonts}
\usepackage{hhline}
\usepackage{cite}
\usepackage{graphicx}
\usepackage{psfrag}
\usepackage{subfigure}
\usepackage{stfloats}
\usepackage{amsmath}   
\usepackage{array}
\newcommand{\rul}{\rule{0pt}{20pt}}

\newtheorem{Lm}{Lemma}

\newtheorem{remark}{Remark}
\newtheorem{exmp}{Example}

\newcommand{\tr}[1]{\mathop{{\rm \bf tr}\left[#1\right]}\nolimits}
\newcommand{\E}[1]{\mathop{{\rm \bf E}\left\{#1\right\}}\nolimits}

\begin{document}
\title{State Sensitivity Evaluation within UD based Array Covariance Filters}

\author{J.V.~Tsyganova and M.V.~Kulikova
\thanks{Manuscript received ??; revised ??. The work of the first author is supported by the Ministry of Education and Science of the Russian Federation (grant number 1.919.2011). The second author thanks the support of Portuguese National Fund ({\it Funda\c{c}\~{a}o para a
Ci\^{e}ncia e a Tecnologia})  within the scope of projects PEst-OE/MAT/UI0822/2011 and SFRH/BPD/64397/2009.}
\thanks{The first author is with Ulyanovsk State University,  Str. L. Tolstoy 42, 432017 Ulyanovsk, Russian Federation. The second author is with   Technical University of Lisbon, Instituto Superior T\'{e}cnico,
           CEMAT (Centro de Matem\'{a}tica e Aplica\c{c}\~{o}es),
          Av. Rovisco Pais 1,  1049-001 Lisboa, Portugal; Emails:
TsyganovaJV@gmail.com; maria.kulikova@ist.utl.pt}}

\markboth{PREPRINT}{}

\maketitle

\begin{abstract}
This technical note addresses the UD factorization based Kalman filtering (KF) algorithms. Using
this important class of numerically stable KF schemes, we extend its functionality and develop an elegant and simple method for computation of sensitivities  of the system state to unknown parameters required in a variety of applications. For instance, it can be used for efficient
calculations in sensitivity analysis and in gradient-search optimization algorithms for the maximum likelihood estimation. The new theory presented in this technical note is a solution to the problem formulated by Bierman {\it et al.} in~\cite{Bierman1990}, which has been open since 1990s. As in the cited paper, our method avoids the standard approach based on the conventional KF (and its derivatives with respect to unknown system parameters) with its inherent numerical instabilities and, hence, improves the robustness of computations against roundoff errors.
\end{abstract}


\IEEEpeerreviewmaketitle

\section{Introduction}

Linear discrete-time stochastic state-space models, with associated Kalman filter (KF),  have been extensively used in practice. Application of the KF  assumes a complete {\it a priori} knowledge of the state-space model parameters, which is a rare case.  As mentioned in~\cite{Sarkka2009}, the classical way of solving the problem of {\it uncertain} parameters is to use {\it adaptive filters} where the model parameters are estimated together with the dynamic state. This requires determination of  sensitivities
of the system state to unknown parameters. Other applications with similar requirements
arise, for instance, in the field of optimal input design~\cite{Mehra1974,Gupta1974} {\it etc}. 

Straight forward differentiation of the KF equations is a direct
approach to  compute the state sensitivities to unknown parameters.
This leads to a set of vector equations, known as the {\it filter sensitivity
equations} and a set of matrix equations,
known as the {\it Riccati-type sensitivity equations}. The main disadvantage of the standard approach is the
problem of numerical instability of the conventional KF (see, for instance, discussion in~\cite{VerhaegenDooren1986}).
The alternative approach can be found in, so-called, square-root (SR)
 algorithms developed to deal with the problem of the numerical instability. Among all existing implementation methods the array SR filters are currently the most preferable for practical implementations. Such methods have the
property of better conditioning and reduced dynamical range. They also imply utilization of numerically stable orthogonal
transformations for each recursion step that leads to more robust
computations (see \cite[Chapter~12]{KailathSayed2000} for an
extended explanation). Since the appearance of the KF in 1960s a number of array SR filters have been developed for the recursive KF~\cite{GrewalAndrews2001,KaminskiBryson1971,Morf1974,Morf1975,ParkKailath1995,Sayed1994} and smoothing formulas~\cite{ParkKailath1995s:1,ParkKailath1995s:2,ParkKailath1996,Gibbs2011}. Recently, some array SR methods have been designed for the $H^{\infty}$ estimation in~\cite{Hassibi2000}.
Here, we deal with the array UD factorization based filters; see~\cite{Bierman1977}. They represent an important class of numerically stable KF implementation methods.

For the problem of the sensitivities computation, the following question arises: would it be possible to update the {\it sensitivity equations} in terms of the variables that appear naturally in the mentioned numerically favored KF algorithms? The first attempt to answer this question belongs to Bierman {\it et al.} In~\cite{Bierman1990} the authors  proposed an elegant method that naturally extends the square-root
information filter (SRIF), developed by Dyer and McReynolds~\cite{Dyer1969}, on the case of the log likelihood gradient (Log~LG) evaluation. Later this approach was generalized to the class of covariance-type filters in~\cite{Kulikova2009IEEE}. However, the problem of utilizing the UD based KF algorithms to generate the required quantities has been open since 1990s. In this technical note we propose and justify an elegant and simple solution to this problem. More precisely, we present a new theory that equips any array UD based KF algorithm with a means for simultaneous computation of derivatives of the filter variables with
respect to unknown system parameters. As in~\cite{Bierman1990}, this avoids implementation of
the conventional KF (and its direct differentiation with respect to
unknown system parameters) because of its inherent numerical
instability and, hence, we improve the robustness of computations against
roundoff errors.  The new results can be used, e.g., for efficient
calculations in sensitivity analysis and in gradient-search optimization algorithms for the maximum likelihood estimation of
unknown system parameters.

\section{Problem statement and the Conventional Kalman Filter}
\PARstart{C}{onsider}  the discrete-time linear stochastic system
  \begin{eqnarray}
   x_{k} &=& F x_{k-1}+G w_k,   \quad k \ge 0, \label{eq2.1} \\
    z_k &=& H x_k+v_k  \label{eq2.2}
  \end{eqnarray}
where $x_k \in \mathbb R^n$ and $z_k \in \mathbb R^m$ are,
respectively, the state and the measurement vectors; $k$ is a
discrete time, i.e. $x_k$ means $x(t_k)$.
 The  process noise, $\{w_k\}$, and  the measurement noise, $\{v_k\}$, are Gaussian white-noise processes, with covariance
matrices $Q \ge 0$ and $R > 0$, respectively. All random variables
have known mean values, which we can take without loss of generality
to be zero. The noises  $w_k \in
\mathbb R^q$, $v_k \in \mathbb R^m$ and the initial state  $x_0 \sim {\mathcal N}(0,\Pi_0)$ are taken from mutually independent
Gaussian distributions.

The associated KF yields  the
linear least-square estimate, $\hat x_{k|k-1}$, of the state vector $x_k$ given the measurements
$Z_1^{k-1}=\{z_1,\ldots, z_{k-1}\}$ that can be computed as follows~\cite{KailathSayed2000}:
$$
\hat x_{k+1|k}  =   F \hat x_{k|k-1}+K_{p,k}e_k, \quad e_k =z_k-H\hat x_{k|k-1}
$$
where  $\hat x_{0|-1} = 0$, $e_k \sim {\cal N}\left(0,R_{e,k}\right)$ are innovations of the KF, $K_{p,k}=\E{\hat x_{k+1|k} e_k^T}$ and ${K_{p,k}=K_kR_{e,k}^{-1}}$, ${K_k=FP_{k|k-1}H^T}$, ${R_{e,k}=R+HP_{k|k-1}H^T}$.
The error covariance matrix
$$P_{k|k-1}=\E{ (x_{k}-\hat
x_{k|k-1})(x_{k}-\hat x_{k|k-1})^{T}}$$
 satisfies the difference
Riccati equation
$$
P_{k+1|k}  = FP_{k|k-1}F^T+GQG^T - K_{p,k}R_{e,k}K_{p,k}^T,
\; P_{0|-1} = \Pi_0 > 0.
$$
In practice, the matrices characterizing the dynamic model are often known up
 to certain parameters. Hence, we move to a more complicated problem. Assume that
system~(\ref{eq2.1})--(\ref{eq2.2}) is parameterized by a vector of
unknown system parameters ${\theta \in \mathbb R^p}$ that needs to be estimated. This means that the entries of the
matrices ${F \in \mathbb R^{n\times n}}$, ${G \in \mathbb R^{n\times
q}}$, ${H \in \mathbb R^{m\times n}}$, ${Q \in \mathbb R^{q\times q}}$,
${R \in \mathbb R^{m\times m}}$ and
${\Pi_0 \in \mathbb R^{n\times n}}$ are functions of $\theta$. For the sake of simplicity we will suppress
the corresponding notations below, i.e instead of $F(\theta)$,
$G(\theta)$, $H(\theta)$ {\it etc.} we will write $F$, $G$, $H$ {\it
etc.}

Solving the parameter estimation problem by the method of maximum
likelihood requires maximization of the likelihood function (LF)
with respect to unknown system parameters. It is often done by using
the gradient approach\footnote{More about computational aspects of maximum likelihood
estimation, different gradient-based nonlinear programming methods
and their applicability to maximum likelihood estimation could be
found in~\cite{Gupta1974}.}  where the computation of the likelihood
gradient (LG) is necessary. For the state-space
system~(\ref{eq2.1})-(\ref{eq2.2}), the LF and LG evaluation demands an implementation
of the KF and, so-called, ``differentiated'' KF to determine the
sensitivities of the system state  to the unknown system parameters,
as explained, for instance,
in~\cite{Gupta1974,
Bierman1990,Kulikova2009IEEE}. More precisely, the LG computation leads to a set of $p$ {\it filter sensitivity
equations} for computing $\partial{\hat
x_{k|k-1}}/\partial{\theta}$ and a set of $p$ matrix {\it Riccati-type sensitivity equations} for computing
$\partial{P_{k|k-1}}/\partial{\theta}$.  Our goal is to avoid the direct
differentiation of the conventional KF equations because of their
inherent numerical instability. Alternatively, we are going to apply a numerically
favored array UD based filter. However, for the task of sensitivities
computation, we have to augment this numerical scheme with a
procedure for robust evaluation of the derivatives of the UD filter variables with
respect to unknown system parameters.

\section{UD Based Array Covariance Filter \label{UD-filter}}

{\it Notations to be used:}  Let $D$ denotes a diagonal matrix; $U$ and $L$ are, respectively,
unit upper  and lower triangular matrices;
$\bar U$ and $\bar L$ are, respectively, strictly upper and lower triangular matrices.  We use Cholesky decomposition of the form ${A=A^{T/2} A^{1/2}}$, where $A^{1/2}$ is an upper triangular matrix. For convenience we will
write ${A^{-1/2}=(A^{1/2})^{-1}}$, ${A^{-T/2}=(A^{-1/2})^T}$ and
$\partial A/ \partial \theta_i$ implies the partial derivative of the
matrix $A$ with respect to the $i$th component of $\theta$ (we assume that the entries of $A$ are differentiable functions of a parameter $\theta$).
 Besides, we use modified Cholesky decomposition of the form ${A=U_{A}D_{A}U^T_{A}}$.

The first UD based filter was developed by Bierman~\cite{Bierman1977}.
Later,  Jover and Kailath~\cite{JoverKailathSayed1986} have presented the advantageous array form for the UD filters.
In this technical note we use the  UD based array covariance filter (UD-aCF) from~\cite[p.~261]{GrewalAndrews2001}.
 First, one have to set the initial values ${\hat x_{0|-1}=0}$, ${P_{0|-1}=\Pi_0>0}$ and use the modified Cholesky
 decomposition to compute the factors $\{ U_{\Pi_{0}}, D_{\Pi_{0}}\}$, $\{ U_{R}, D_{R}\}$, $\{ U_{Q}, D_{Q}\}$. Then,
 we recursively update the $\{ U_{P_{k+1|k}}, D_{P_{k+1|k}}\}$ as follows ($k=1, \ldots, N$):
  given a pair of the pre-arrays $\{ {\mathcal A_k}, {\mathcal D_k}\}$
\begin{equation}
\label{adapt_filter:2}
{\mathcal A}_k^T=\left[
\begin{IEEEeqnarraybox}[][c]{c/c/c}
GU_{Q} & FU_{P_{k|k-1}} & 0\\
0 & HU_{P_{k|k-1}} & U_{R}
\end{IEEEeqnarraybox}
\right], \;
{\mathcal D}_k = {\rm diag}\{ D_Q, D_{P_{k|k-1}}, D_R\}
,
\end{equation}
apply the modified weighted Gram-Schmidt (MWGS) orthogonalization \cite{Bjorck1967} of the columns of ${\mathcal A}_k$ with respect to  the weighting matrix ${\mathcal D}_k$ to obtain a pair of the post-arrays $\{ \tilde{\mathcal A}_k, \tilde{\mathcal D}_k\}$
\begin{equation}
\label{adapt_filter:5}
 \tilde{\mathcal A}_k=\left[
\begin{IEEEeqnarraybox}[][c]{c/c}
U_{P_{k+1|k}} & K_{p,k}U_{R_{e,k}}\\
0 & U_{R_{e,k}}
\end{IEEEeqnarraybox}
\right], \;
  \tilde{\mathcal D}_k = {\rm diag}\{D_{P_{k+1|k}}, D_{R_{e,k}}\}
\end{equation}
such that
$$
  {\mathcal A}_k^T=\tilde{\mathcal A}_kB^T_k \quad \mbox{and} \quad {\mathcal A}_k^T{\mathcal D}_k{\mathcal A}_k=\tilde{\mathcal A}_k \tilde{\mathcal D}_k \tilde{\mathcal A}_k^T
$$
 where  ${\mathcal D}_{k} \in {\mathbb R}^{(n+m+q)\times (n+m+q)}$, ${\mathcal A}_k \in {\mathbb R}^{(n+m+q)\times (n+m)}$, $B_k \in {\mathbb R}^{(n+m+q)\times (n+m)}$ is the MWGS transformation that produces the block upper triangular matrix $\tilde{\mathcal A}_k   \in {\mathbb R}^{(n+m)\times (n+m)}$ and diagonal matrix $\tilde{\mathcal D}_k  \in {\mathbb R}^{(n+m)\times (n+m)}$.

The state estimate can be computed as follows:
\begin{equation}
\label{adapt_filter:9}
      \hat x_{k+1|k}=F\hat x_{k|k-1}+\left(K_{p,k}U_{R_{e,k}}\right)\bar e_k
\end{equation}
where $\bar e_{k} =U^{-1}_{R_{e,k}}e_k$, \; $e_k= z_k-H\hat x_{k|k-1}$.

\begin{remark}
\label{remark:1}
We note that the parenthesis in~(\ref{adapt_filter:9}) are used to indicate the quantities
that can be directly read off from the post-arrays.
\end{remark}

Instead of the conventional KF, which is known to be numerically
unstable, we wish to utilize stable UD-aCF filter presented above to compute the Log LF:
\begin{equation}
\label{LLF:conventional}
{\mathcal L}_{\theta}\left(Z_1^N\right)=-\frac{Nm}{2}\ln(2\pi) -
\frac{1}{2} \sum \limits_{k=1}^N \left\{
 \ln\left(\det R_{e,k}\right)+ e_k^T R_{e,k}^{-1}e_k \right\}
\end{equation}
where $Z_1^N=\{z_1,\ldots, z_N\}$ is $N$-step measurement history
and the innovations, $\{ e_k \}$, $e_k \sim {\cal
N}\left(0,R_{e,k}\right)$, are generated by the discrete-time
KF.

One can easily obtain the expression for Log
LF~(\ref{LLF:conventional}) in terms of the UD-aCF variables:
\begin{equation}
\label{eq:llfUD-CF}
 {\mathcal L}_{\theta}\left(Z_1^N\right)=
-\frac{Nm}{2}\ln(2\pi) - \frac{1}{2} \sum \limits_{k=1}^N \left\{
  \ln\left(\det D_{R_{e,k}}\right)+ \bar e_k^T D^{-1}_{R_{e,k}}\bar e_k
\right\}.
\end{equation}

Let $\theta=[\theta_1,\dots,\theta_p]$ denote the vector of parameters with respect to which the likelihood function is to be differentiated. Then from (\ref{eq:llfUD-CF}),  we have 
$$
\frac{\partial  {\mathcal L}_{\theta}\left(Z_1^N\right)}{\partial \theta_i} =
 - \frac{1}{2} \sum \limits_{k=1}^N \left\{
  \frac{\partial \left[\ln\left(\det D_{R_{e,k}}\right)\right]}{\partial \theta_i}
+\frac{\partial \left[\bar e_k^T D^{-1}_{R_{e,k}}\bar e_k\right]}{\partial \theta_i}
\right\}
$$
where $i=1, \ldots, p$.

Taking into account that the matrix $D_{R_{e,k}}$ is diagonal
and using Jacobi's formula, we obtain the expression for the Log LG evaluation in terms of the UD-aCF variables ($i=1, \ldots, p$):
\begin{eqnarray}
\nonumber
\frac{\partial {\mathcal L}_{\theta}\left(Z_1^N\right)}{\partial \theta_i} & = &
-\frac{1}{2}\sum \limits_{k=1}^N \left\{
\tr { \frac{\partial D_{R_{e,k}}}{\partial
\theta_i}D^{-1}_{R_{e,k}}} +
2\frac{\partial \bar e_k^T}{\partial \theta_i}D^{-1}_{R_{e,k}}\bar e_k \right.\\
\label{grad-UD-CF}
& & \left. -\bar e_k^TD^{-2}_{R_{e,k}}\frac{\partial D_{R_{e,k}}}{\partial \theta_i}\bar e_k
\right\}.
\end{eqnarray}

Our goal is to compute Log
LF~(\ref{eq:llfUD-CF}) and Log LG~(\ref{grad-UD-CF}) by using the UD-aCF variables. As can be seen, the elements $\bar e_k$ and $D_{R_{e,k}}$ are readily available from UD based filter~(\ref{adapt_filter:2})--(\ref{adapt_filter:9}). Hence, our aim is to explain how the last two terms, i.e.
${\partial \bar e_k}/{\partial \theta_i}$ and ${\partial D_{R_{e,k}}}/{\partial \theta_i}$, $i=1, \ldots, p$, can be computed using quantities available from the UD-aCF algorithm.

\section{Main Results \label{sec:main}}

In this section, we present a simple and convenient technique that naturally augments any array UD based filter for computing derivatives of the filter variables.
To begin constructing the method, we note that each iteration of the UD based implementation has the following form: given a pair of the pre-arrays $\{ A, D_{w}\}$, compute a pair of the post-arrays $ \{U, D_{\beta} \}$ by means of the MWGS orthogonalization, i.e.
\begin{equation}
\label{assume:1}
 A^T=UB^T \quad \mbox{and} \quad A^TD_{w}A=UD_{\beta}U^T
\end{equation}
 where $A \in {\mathbb R}^{r\times s}$, $r>s$ and $B  \in {\mathbb R}^{r\times s}$ is the MWGS transformation that produces the block upper triangular matrix $U  \in {\mathbb R}^{s\times s}$. The diagonal matrices $D_w \in {\mathbb R}^{r\times r}$, $D_\beta \in {\mathbb R}^{s\times s}$ satisfy $B^TD_wB=D_\beta$ and  $D_w>0$ (see~\cite[Lemma~VI.4.1]{Bierman1977} for an extended explanation).

\begin{Lm}
\label{lemma:1}  Let entries of the pre-arrays $A$, $D_w$ in~(\ref{assume:1}) be known differentiable functions of a parameter $\theta$.
 Consider the transformation in~(\ref{assume:1}). Given the derivatives of the pre-arrays $A'_{\theta}$ and $(D_w)'_\theta$, the following formulas calculate the corresponding derivatives of the post-arrays:
\begin{equation}
 \label{lemma2:eq:1}
 U'_{\theta} = U\left({\bar L}^T_0+{\bar U}_0+{\bar U}_2\right)D^{-1}_\beta\quad \mbox{and} \quad
\left(D_\beta\right)'_\theta=2D_0+D_2
\end{equation}
where the quantities $\bar L_{0}$, $D_{0}$, $\bar U_{0}$ are, respectively, strictly lower triangular, diagonal and strictly upper triangular parts of the matrix product $B^TD_wA'_\theta U^{-T}$. Besides, $D_2$ and $\bar U_{2}$ are diagonal and strictly upper triangular parts of the product $B^T(D_w)'_\theta B$, respectively.
\end{Lm}

\begin{proof}
For the sake of simplicity we transpose the first equation in~(\ref{assume:1}) to obtain $A=BL$. As shown in~\cite{JoverKailathSayed1986}, the matrix $B$ can be represented in the form of $B=D_w^{-1/2}TD_\beta^{1/2}$ where $T$ is the matrix with orthonormal columns, i.e. $T^TT=I$, and $I$ is an identity matrix. Next, we define $B^+=D_\beta^{-1/2}T^{\,T}D_w^{1/2}$ and note that $B^+B=I$. The mentioned matrices $B$, $B^+$ exist since the  $D_w$, $D_\beta$ are invertible. Indeed, the diagonal matrix $D_w$ is a positive definite matrix, i.e. $D_w>0$, and $D_\beta$ satisfies $B^TD_wB=D_\beta$ (see~\cite[Lemma~VI.4.1]{Bierman1977} for further details).

Multiplication of both sides of the equality $A=BL$ by the matrix $B^+$  and, then, their differentiation with respect to $\theta$ yield
$$B^+A'_\theta+\left(B^+\right)'_\theta A=L'_\theta.$$ Therefore,
$B^+A'_\theta L^{-1}+\left(B^+\right)'_\theta A L^{-1}=L'_\theta L^{-1}$ or
\begin{equation}
\label{eq:main13}
B^+A'_\theta L^{-1}+\left(B^+\right)'_\theta B=L'_\theta L^{-1}.
\end{equation}

Now, we consider the product $\left(B^+\right)'_\theta B$ in~(\ref{eq:main13}). For any diagonal matrix $D_\beta$, we have
$\left(D_\beta^{-1/2}\right)'_\theta = - D_\beta^{-1/2} \left(D_\beta^{1/2}\right)'_\theta  D_\beta^{-1/2}$,
and
$\left(D_\beta\right)'_\theta = 2D_\beta^{1/2}\left( D_\beta^{1/2}\right)'_{\theta}$.

Taking into account that $B^+=D_\beta^{-1/2}T^{\,T}D_w^{1/2}$, we further obtain
\begin{eqnarray}
\nonumber
\left(B^+\right)'_\theta B & = & \left[\left(D_\beta^{-1/2}\right)'_\theta T^TD_w^{1/2}+D_\beta^{-1/2}\left(T'_\theta\right)^T D_w^{1/2}\right.\\
\nonumber
& & \left.+D_\beta^{-1/2}T^T\left(D_w^{1/2}\right)'_\theta\right]D_w^{-1/2}T D_\beta^{1/2}\\
\nonumber
 & = & \left(D_\beta^{-1/2}\right)'_\theta D_\beta^{1/2}
+D_\beta^{-1/2}\left(T'_\theta\right)^T T D_\beta^{1/2}\\
\nonumber
& & +D_\beta^{-1/2}T^T \left(D_w^{1/2}\right)'_\theta D_w^{-1/2}T D_\beta^{1/2} \\
\nonumber
 & = & -\frac12 D^{-1}_\beta \left(D_\beta\right)'_\theta+D_\beta^{-1/2}\left(T'_\theta\right)^T T D_\beta^{1/2}\\
\label{eq:main4}
& & +D_\beta^{-1/2}T^T \left(D_w^{1/2}\right)'_\theta D_w^{-1/2}T D_\beta^{1/2}.
\end{eqnarray}

Next, we study the last term in~(\ref{eq:main4}). First, taking into account that $D_w$ is a diagonal matrix, we derive
$$
\left(D_w^{1/2}\right)'_\theta D_w^{-1/2}  =  \frac12 D_w^{-1/2}\left(D_w\right)'_\theta D_w^{-1/2}=\frac12 D_w^{-1}\left(D_w\right)'_\theta.
$$
Thus, the above formula and an equality $T=D_w^{1/2}BD_\beta^{-1/2}$ yield
\begin{equation}\label{eq:main9}
D_\beta^{-1/2}T^T \left(D_w^{1/2}\right)'_\theta D_w^{-1/2}T D_\beta^{1/2} =
\frac12 D_\beta^{-1}B^T\left(D_w\right)'_\theta B.
\end{equation}

Furthermore, we show that the term $\left(T'_\theta\right)^T T$ required in~(\ref{eq:main4}) is a skew symmetric matrix. For that, we differentiate both sides of the formula ${T^TT=I}$ with respect to $\theta$ and arrive at
${\left(T'_\theta\right)^T T + T^T T'_\theta=0}$,
 or in the equivalent form $$\left(T'_\theta\right)^T T=-\left(\left(T'_\theta\right)^T T\right)^T.$$ The latter implies that the matrix $\left(T'_\theta\right)^T T$ is skew symmetric and can be presented as a difference of two matrices, i.e. $\left(T'_\theta\right)^T T={\bar U}_1^T-{\bar U}_1$
where $\bar U_1$ is a strictly upper triangular matrix.

Final substitution of $\left(T'_\theta\right)^T T={\bar U}_1^T-{\bar U}_1$  and (\ref{eq:main9}) into~(\ref{eq:main4}) and, then, into~(\ref{eq:main13}) yields
\begin{eqnarray}
\nonumber
L'_\theta L^{-1} & = & B^+A'_\theta L^{-1}+
 D_\beta^{-1}\left[-\frac12 \left(D_\beta\right)'_\theta \right.\\
\label{eq:main14}
& & \left.+D_\beta^{1/2}({\bar U}_1^T-{\bar U}_1)D_\beta^{1/2}
 +\frac12 B^T\left(D_w\right)'_\theta B\right].
\end{eqnarray}

Next, from $B^+=D_\beta^{-1/2}T^T D_w^{1/2}$ and $T=D_w^{1/2}BD_\beta^{-1/2}$, we derive $B^+=D_\beta^{-1}B^TD_w$. Further multiplication of both sides of~(\ref{eq:main14}) by $D_\beta$ yields
\begin{eqnarray}
\nonumber
D_\beta L'_\theta L^{-1} & = & B^T D_wA'_\theta L^{-1}-\frac12 \left(D_\beta\right)'_\theta
+D_\beta^{1/2}({\bar U}_1^T-{\bar U}_1)D_\beta^{1/2}\\
\label{eq:main15}
& &
+\frac12 B^T\left(D_w\right)'_\theta B.
\end{eqnarray}

Now, let us discuss equation~(\ref{eq:main15}) in details. We note that the term $B^T D_wA'_\theta L^{-1}$ in~(\ref{eq:main15}) is a full matrix. Hence, it can be represented in the form of $B^TD_wA'_\theta L^{-1}={\bar L}_0+D_0+{\bar U}_0$ where $\bar L_0$, $D_0$ and $\bar U_0$ are, respectively, strictly lower triangular, diagonal and strictly upper triangular parts of $B^TD_wA'_\theta L^{-1}$. Next, the matrix product $B^T\left(D_w\right)'_\theta B$ in~(\ref{eq:main15}) is a symmetric matrix and, hence, it has the form $B^T\left(D_w\right)'_\theta B={\bar U}_2^T+D_2+{\bar U}_2$
where $D_2$ and $\bar U_2$ are, respectively, diagonal and strictly upper triangular parts of $B^T\left(D_w\right)'_\theta B$. Thus, equation~(\ref{eq:main15}) can be represented in the following form:
\begin{eqnarray}
\nonumber
D_\beta L'_\theta L^{-1} & = & \underbrace{{\bar L}_0+D_0+{\bar U}_0}_{B^T D_wA'_\theta L^{-1}}-\frac12 \left(D_\beta\right)'_\theta+D_\beta^{1/2}({\bar U}_1^T-{\bar U}_1)D_\beta^{1/2}\\
\label{eq:main16}
& & +\frac12 \underbrace{\left({\bar U}_2^T+D_2+{\bar U}_2\right)}_{B^T\left(D_w\right)'_\theta B}.
\end{eqnarray}

Next, we note that the left-hand side matrix in~(\ref{eq:main16}), i.e.  the matrix $D_\beta L'_\theta L^{-1}$, is a strictly lower triangular (since $L$ is a unit lower triangular matrix). Hence, the matrix on the right-hand side of~(\ref{eq:main15}) should be also a  strictly lower triangular matrix. In other words, the strictly upper triangular and diagonal parts of the matrix  on the right-hand side of~(\ref{eq:main16}) should be zero. Hence, the formulas ${D_0-\frac12 \left(D_\beta\right)'_\theta +\frac12 D_2=0}$,
${\bar U_0-D_\beta^{1/2}{\bar U}_1D_\beta^{1/2}+\frac12 \bar U_2=0}$ imply that
\begin{equation}\label{eq:main17}
\left(D_\beta\right)'_\theta = 2D_0+D_2\quad \mbox{and} \quad
D_\beta^{1/2}{\bar U}_1D_\beta^{1/2} = \bar U_0+\frac12 \bar U_2.
\end{equation}
Clearly, the first equation in~(\ref{eq:main17}) is exactly the second formula in~(\ref{lemma2:eq:1}). Eventually, the substitution of both formulas in~(\ref{eq:main17}) into~(\ref{eq:main16}) validates the first relation in~(\ref{lemma2:eq:1}). More precisely, it results in
$
{D_\beta L'_\theta L^{-1}={\bar L}_0+{\bar U}^T_0+{\bar U}_2^T}$, and the latter formula means that
$
{L'_\theta =D^{-1}_\beta \left({\bar L}_0+{\bar U}^T_0+{\bar U}_2^T\right)L}
$
where $L$ stands for $U^T$. This completes the proof of Lemma~\ref{lemma:1}.
\end{proof}

We see that the proposed computational procedure utilizes only the pre-arrays $A$, $D_{w}$,
their derivatives with respect to the unknown system parameters, the
post-arrays  $U$, $D_{\beta}$ and the  MWGS orthogonal transformation in
order to compute the derivatives of the post-arrays.

\section{UD Filter Sensitivity Evaluation \label{sec:method}}

Further, we suggest a general computation scheme that naturally
extends any array UD based filtering algorithm to the above-mentioned
derivative evaluation. We stress that our method allows the {\it
filter and Riccati-type sensitivity} equations to be updated in
terms of stable array UD filters.

To illustrate the proposed approach, we apply Lemma~1 to the UD-aCF algorithm presented in Section~\ref{UD-filter} with $r=m+n+q$, $s=m+n$,  and the following pre-, post-arrays from~(\ref{adapt_filter:2}), (\ref{adapt_filter:5}):
$$
D_{w} = {\mathcal D}_k, \quad A={\mathcal A}_k \quad \mbox{and} \quad
D_{\beta}=\tilde{\mathcal D}_k, \quad U=\tilde{\mathcal A}_k .
$$
\subsection{Summary of Computations. Extended UD based KF scheme.}
\begin{tabbing}
{\it Step~0.\;} Set a current value of $\hat \theta$, \\
{\it Step~1.\;} Evaluate  \= $\hat F=F\left|_{\hat \theta}\right.$,
$\hat G =G\left|_{\hat \theta} \right.$, $\hat H = H\left|_{\hat
\theta} \right.$,  \\
 {\it \phantom{Step~1.}} $\hat Q = Q\left|_{\hat \theta} \right.$, $\hat R= R\left|_{\hat\theta} \right.$, $\hat \Pi_0 = \Pi_0\left|_{\hat \theta}\right.$;\\
{\it \phantom{Step~1.}}
$\displaystyle\frac{\partial \hat
F}{\partial \theta_i}=\left.\displaystyle\frac{\partial F}{\partial
\theta_i}\right|_{\hat \theta}$,
$\displaystyle\frac{\partial \hat
G}{\partial \theta_i}=\left.\displaystyle\frac{\partial G}{\partial
\theta_i}\right|_{\hat \theta}$, $\displaystyle\frac{\partial \hat
H}{\partial \theta_i}=\left.\displaystyle\frac{\partial H}{\partial
\theta_i}\right|_{\hat \theta}$,\\
{\it \phantom{Step~1.}}
 $\displaystyle\frac{\partial \hat
Q}{\partial \theta_i}=\left.\displaystyle\frac{\partial Q}{\partial
\theta_i}\right|_{\hat \theta}$,
$\displaystyle\frac{\partial \hat
R}{\partial \theta_i}=\left.\displaystyle\frac{\partial R}{\partial
\theta_i}\right|_{\hat \theta}$;\\
{\it Step~2.\;} Set  the initial conditions:
$P_{0|-1}=\hat \Pi_{0}$, $ \hat x_{0|-1}= 0$  and \\
{\it \phantom{Step~2.}}
 $\displaystyle\frac{\partial P_{0|-1}}{\partial
\theta_i}=\displaystyle\frac{\partial \hat \Pi_0}{\partial
\theta_i}$,
  $\displaystyle\frac{\partial \hat
x_0}{\partial \theta_i}=0$. \\
{\it Step 3.\;} Use the modified Cholesky decomposition to find \\
{\it \phantom{Step 2.\;}}
$\{ U_{\hat \Pi_{0}}, D_{\hat \Pi_{0}}\}$, $\{ U_{\hat R}, D_{\hat R}\}$, $\{ U_{\hat Q}, D_{\hat Q}\}$ \\
{\it \phantom{Step 2.\;}} and   
$\left\{\displaystyle\frac{\partial U_{\hat \Pi_{0}}}{\partial \theta_i}, \displaystyle\frac{\partial D_{\hat \Pi_{0}}}{\partial \theta_i}\right\}$,
$\left\{\displaystyle\frac{\partial U_{\hat R}}{\partial \theta_i}, \displaystyle\frac{\partial D_{\hat R}}{\partial \theta_i}\right\}$,
$\left\{\displaystyle\frac{\partial U_{\hat Q}}{\partial \theta_i}, \displaystyle\frac{\partial D_{\hat Q}}{\partial \theta_i}\right\}$; \\
{\it Step~4.$\;$} For \=$t_k$, $k=0,\ldots,N-1$, do\\
{\it Step~5.} \>Apply \=a numerically stable array UD based filter:\\
{\it Step~6.} \> \> Form the pre-arrays $A$, $D_w$ by using \\
{\it \phantom{Step~6.}} \> \> the matrices from Step~1; \\
{\it Step~7.} \> \> Compute the post-arrays $U$, $D_{\beta}$  using \\
{\it \phantom{Step~7.}} \> \> the MWGS algorithm.  Save $B$, $U$, $D_{\beta}$;\\
{\it Step~8.} \> \> Compute the state estimate $\hat x_{k+1|k}$ \\
{\it \phantom{Step~8.}} \> \> according to (\ref{adapt_filter:9}).\\
{\it Step~9.} \> Apply  \= the designed derivative computation method  \\
{\it \phantom{Step~9.}} \> (for each $\theta_i: i=1, \ldots, p$):\\
{\it Step~10.} \> \> Form the derivatives $\displaystyle\frac{\partial A}{\partial \theta_i}$, $\displaystyle\frac{\partial D_w}{\partial \theta_i}$ by using \\
{\it \phantom{Step~10.}} \> \> the matrices from Steps~1 and 3; \\
{\it Step~11.} \> \> Calculate $B^TD_w \displaystyle\frac{\partial A}{\partial \theta_i} U^{-T}$ \\
{\it \phantom{Step~11.}} \> \> (use the quantities from Steps~6 and 10); \\
{\it Step~12.} \> \> Split it into $\bar L_0(i)$, $D_0(i)$ and $\bar U_0(i)$. \\
{\it \phantom{Step~12.}} \> \> Save $\{ \bar L_0(i), D_0(i), \bar U_0(i) \}$;\\
{\it Step~13.} \> \> Calculate $B^T \displaystyle\frac{\partial D_w}{\partial \theta_i} B$ \\
{\it \phantom{Step~13.}} \> \> (use the quantities from Steps~7, 10); \\
{\it Step~14.} \> \> Split it into $\bar U^T_2(i)$, $D_2(i)$ and $\bar U_2(i)$. \\
{\it \phantom{Step~14.}} \> \> Save $\{ D_2(i), \bar U_2(i) \}$;\\
{\it Step~15.} \> \> Find $\displaystyle\frac{\partial
U}{\partial \theta_i}=U\left[{\bar L}^T_0(i)+{\bar U}_0(i)+{\bar U}_2(i)\right]D^{-1}_\beta$ \\
{\it \phantom{Step~15.}} \> \> (use Steps~7, 12, 14);\\
{\it Step~16.} \> \> Evaluate \= $\displaystyle\frac{\partial
D_{\beta}}{\partial \theta_i}  = 2D_0(i)+D_2(i)$ \\
{\it \phantom{Step~16.}} \> \> (use the saved values from Steps~12, 14);\\
{\it Step~17.} \> \> Evaluate the state sensitivity as \\
\> \>
$
\displaystyle\frac{ \partial \hat x_{k+1|k}}{\partial \theta_i} = \frac{ \partial \left( F\hat x_{k|k-1}\right)}{\partial \theta_i} +\frac{ \partial \left(K_{p,k}U_{R_{e,k}}\right)}{\partial \theta_i}\bar e_k$ \\
{\it \phantom{Step~17.}} \> \>
\phantom{$\displaystyle\frac{ \partial \hat x_{k+1|k}}{\partial \theta_i} =$}
$+\left(K_{p,k}U_{R_{e,k}}\right)\displaystyle\frac{\partial \bar e_{k}}{\partial \theta_i}$;\\
{\it Step~18.\quad}  End.
\end{tabbing}

\begin{remark}
The new  approach naturally extends any UD based KF implementation on the filter sensitivities evaluation. Additionally, this allows Log
LF~(\ref{eq:llfUD-CF}) and Log LG~(\ref{grad-UD-CF}) to be computed simultaneously. Hence, such
methods are ideal for simultaneous state estimation and parameter
identification.
\end{remark}

\begin{table*}
\renewcommand{\arraystretch}{1.3}
\caption{Illustrative calculation for Example~\ref{ex:1:1}}  \label{tab:matrix}
\centering
{\scriptsize
\begin{tabular}{|l|l|}
\hline
 Steps~5, 6. & We are given
$A=\left[
\begin{IEEEeqnarraybox}[][c]{c/c}
{\theta^5}/20 & \theta^4/8 \\
{\theta^4}/8 & \theta^3/3 \\
{\theta^3}/{6} & {\theta^2}/{2}
\end{IEEEeqnarraybox}
\right]$, $D_w = \left[
\begin{IEEEeqnarraybox}[][c]{c/c/c}
\theta & 0 & 0 \\
0 & \theta^2 & 0 \\
0 & 0 & \theta^3
\end{IEEEeqnarraybox}
\right]$. Hence, $\left.A\right|_{\theta=2}=\left[
\begin{IEEEeqnarraybox}[][c]{c/c}
8/5 & 2\\
2 & 8/3 \\
4/3 & 2
\end{IEEEeqnarraybox}
\right]$, $\left.D_w\right|_{\theta=2}=\left[
\begin{IEEEeqnarraybox}[][c]{c/c/c}
2& 0 & 0\\
0 & 4 & 0  \\
0 & 0 & 8
\end{IEEEeqnarraybox}
\right]$ \\ \rul
Step~7. & Compute the post-arrays
$U=\left[
\begin{IEEEeqnarraybox}[][c]{c/c}
 1.0000   & 0.7169 \\
         0  &  1.0000
\end{IEEEeqnarraybox}
\right]$, $D_{\beta}=\left[
\begin{IEEEeqnarraybox}[][c]{c/c}
 0.1672   & 0 \\
         0  &  68.4444
\end{IEEEeqnarraybox}
\right]$  where
$B=\left[
\begin{IEEEeqnarraybox}[][c]{c/c}
  0.1662   & 2.0000 \\
  0.0883  &  2.6667 \\
  -0.1004  &  2.0000
 \end{IEEEeqnarraybox}
\right]$. \\ \rul
 Steps~9, 10. & We are given
$A'_{\theta}=\left[
\begin{IEEEeqnarraybox}[][c]{c/c}
{\theta^4}/4 & \theta^3/2\\
{\theta^3}/2 & \theta^2\\
\theta^2/2 &  \theta
\end{IEEEeqnarraybox}
\right]$, $\left( D_w \right)'_{\theta}=\left[
\begin{IEEEeqnarraybox}[][c]{c/c/c}
1 & 0 & 0 \\
0 & 2\theta & 0 \\
0 & 0 & 3\theta^2
\end{IEEEeqnarraybox}
\right]$. So,  $\left.A'_{\theta}\right|_{\theta=2}=\left[
\begin{IEEEeqnarraybox}[][c]{c/c}
4 & 4 \\
4 & 4 \\
2 & 2
\end{IEEEeqnarraybox}
\right]$, $\left. \left( D_w \right)'_{\theta}\right|_{\theta=2}=\left[
\begin{IEEEeqnarraybox}[][c]{c/c/c}
1 & 0 & 0 \\
0 & 4 & 0 \\
0 & 0 & 12
\end{IEEEeqnarraybox}
\right]$.  \\ \rul
Steps~11, 12. & Compute
$B^TD_w A'_{\theta} U^{-T}$. Save $\bar L_0(1)=\left[
\begin{IEEEeqnarraybox}[][c]{c/c}
   0 &  0 \\
   25.6693 &  0
\end{IEEEeqnarraybox}
\right]$, $D_0(1) =
\left[\begin{IEEEeqnarraybox}[][c]{c/c}
  0.3216  & 0 \\
   0 &  90.6667
\end{IEEEeqnarraybox}
\right]$,
$\bar U_0(1) =
\left[\begin{IEEEeqnarraybox}[][c]{c/c}
  0 & 1.1359 \\
   0 &  0
\end{IEEEeqnarraybox}
\right]$.\\ \rul
 Step~13, 14. & Evaluate
$B^T\left(D_w\right)'_{\theta}B$. Save $D_2(1) =
\left[\begin{IEEEeqnarraybox}[][c]{c/c}
  0.1799  & 0 \\
   0 &  80.4444
\end{IEEEeqnarraybox}
\right]$,
$\bar U_2(1) =
\left[\begin{IEEEeqnarraybox}[][c]{c/c}
  0 & -1.1359 \\
   0 &  0
\end{IEEEeqnarraybox}
\right]$.\\ \rul
Step~15, 16. & Finally, the results that we are looking for are
$\left.U'_{\theta}\right|_{\theta=2}= \left[
\begin{IEEEeqnarraybox}[][c]{c/c}
  0    & 0.3750 \\
         0   & 0
\end{IEEEeqnarraybox}
\right]$,
$\left.\left(D_{\beta}\right)'_{\theta}\right|_{\theta=2}= \left[
\begin{IEEEeqnarraybox}[][c]{c/c}
  0.8231    & 0 \\
         0   & 261.7778
\end{IEEEeqnarraybox}
\right]$. \\[6pt]
\hline
\end{tabular}
}
\end{table*}

\section{Numerical Examples\label{experiments}}

\subsection{Simple Test Problem}

First, we would like to check our theoretical derivations presented in Lemma~1. To do so,
we apply the proposed UD based computational scheme to the following simple test problem.

\begin{exmp}
\label{ex:1:1} For the given pre-arrays
\begin{equation}\label{A:1}
 A= \left[
\begin{IEEEeqnarraybox}[][c]{c/c}
{\theta^5}/20 \; & \;\theta^4/8 \\
{\theta^4}/8 \; & \;\theta^3/3 \\
{\theta^3}/{6} \; & \; {\theta^2}/{2}
\end{IEEEeqnarraybox}
\right] \mbox{ and } D_w = \left[
\begin{IEEEeqnarraybox}[][c]{c/c/c}
\theta & 0 & 0 \\
0 & \theta^2 & 0 \\
0 & 0 & \theta^3
\end{IEEEeqnarraybox}
\right],
\end{equation}
compute the derivatives of the post-arrays $U'_{\theta}$, $\left(D_{\beta}\right)'_{\theta}$, say, at ${\theta=2}$ where  $U$, $D_{\beta}$
comes from~(\ref{assume:1}).
\end{exmp}

Since the pre-arrays $A$ and $D_w$ are fixed by (\ref{A:1}), we skip
Steps~0--3 in the above-presented computational scheme. We also skip steps~8 and 17, because we do not need to compute the state estimate (and its derivatives) in this simple test problem. Besides, the unknown parameter $\theta$ is a scalar
value, i.e. ${p=1}$. Next, we show the
detailed explanation of only one iteration step, i.e. we set ${k=1}$ in Step~4. The obtained results are summarized in Table~1.
As follows from~(\ref{assume:1}), we have $A^T D_w A = U D_{\beta} U^T$. Thus, the derivatives
of both sides of the latter formula must also agree. We compute the norm $\left|\left|{(A^T D_w A)}'_{\theta
=2}-{(U D_{\beta} U^T)}'_{\theta=2}\right|\right|_2=5.68\cdot 10^{-14}$. This
confirms the correctness of the calculation and the above theoretical derivations presented in Lemma~1.

\subsection{Application to the Maximum Likelihood Estimation of Unknown System Parameters}

The computational approach presented in this paper can be used for
the efficient evaluation of the Log LF and its gradient required in gradient-search optimization algorithms for the maximum
likelihood estimation of unknown system parameters. To demonstrate this, we apply the proposed UD based algorithm to the problem from aeronautical equipment engineering.

\begin{exmp}
\label{ex:1} Consider a simplified version of the instrument error model for one channel of the Inertial Navigation System (INS) of semi-analytical type given as follows~\cite{Broxmeyer56}:
$$
\left[
\begin{IEEEeqnarraybox}[][c]{c}
\Delta v_x (t_{k+1})\\
\beta (t_{k+1})\\
m_{Ax} (t_{k+1})\\
n_{Gy} (t_{k+1})
\end{IEEEeqnarraybox}
\right]
  =
\left[
\begin{IEEEeqnarraybox}[][c]{c/c/c/c}
1 & -\tau g & \tau & 0\\
\tau/a & 1 & 0 & \tau\\
0 & 0 & b_1 & 0\\
0 & 0 & 0 & 1
\end{IEEEeqnarraybox}
\right]
\left[
\begin{IEEEeqnarraybox}[][c]{c}
\Delta v_x (t_{k})\\
\beta (t_{k})\\
m_{Ax} (t_{k})\\
n_{Gy} (t_{k})
\end{IEEEeqnarraybox}
\right]
+
\left[
\begin{IEEEeqnarraybox}[][c]{c}
0\\
0\\
a_1\\
0
\end{IEEEeqnarraybox}
\right]w (t_k),
$$
$$
z(t_k) =  \Delta v_x(t_k) +v(t_k)
$$
where $w(t_k) \sim {\mathcal N}(0,1)$, $v(t_k) \sim {\mathcal N}(0,0.01)$, $x_0 \sim {\mathcal N}(0, I_4)$ and subscripts $x$, $y$, $A$, $G$ denote ``axis $Ox$'', ``axis $Oy$'', ``Accelerometer'', and ``Gyro'', respectively. The $\Delta v_x$ is the  random error in reading velocity along axis $Ox$ of a gyro-stabled platform (GSP), $\beta$ is the angular error in determining the local vertical, $m_{Ax}$ is the accelerometer reading random error, and $n_{Gy}$ is the gyro constant drift rate. The quantities $a_1 \simeq H_1\sqrt{2\gamma_1\tau}$ and $b_1 \simeq 1-\gamma_1\tau$ depend on the unknown parameter $\gamma_1$ that needs to be estimated.
\end{exmp}

\begin{figure*}
\centerline{\subfigure[The negative Log LF ]{\includegraphics[trim = 5mm 5mm 0mm 5mm, clip, keepaspectratio=true, width=2.8in]{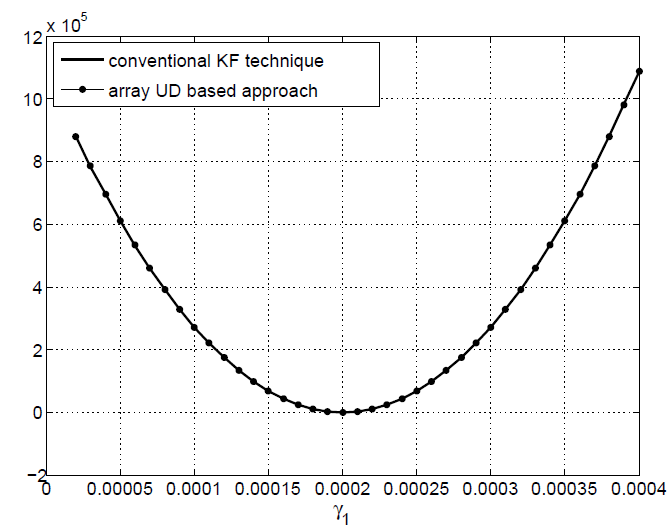}
\label{fig:1:a}}
\hfil
\subfigure[The negative Log LG]{\includegraphics[trim = 5mm 5mm 0mm 5mm, clip, keepaspectratio=true,width=2.8in]{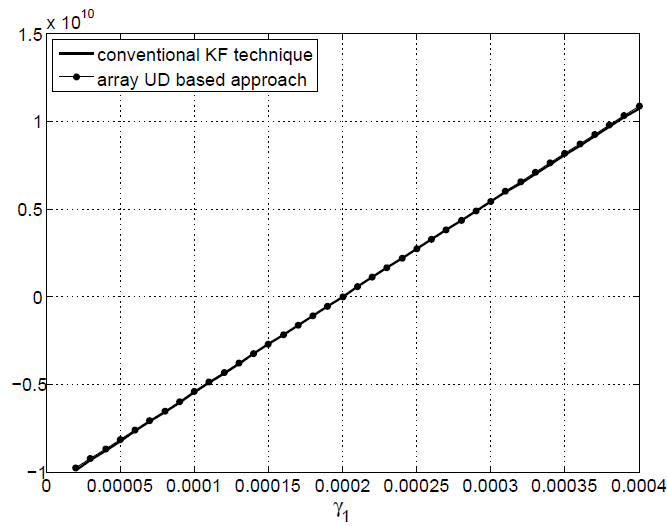}
\label{fig:1:b}}}
\caption{The negative Log LF and its gradient computed by the proposed array UD approach and the conventional KF technique}
\label{fig:1}
\end{figure*}

In our simulation experiment we assume that the {\it true} value of $\gamma_1$ is ${\gamma^*_1=2\cdot 10^{-4}}$. We  compute the negative Log LF and its gradient by the proposed array UD based scheme and then compare the results to those produced by the classical, i.e. the conventional KF approach. The computations are done on the interval $[10^{-5}, 4\cdot 10^{-4}]$ with the step $10^{-5}$. The outcomes of these experiments are illustrated by Fig.~\ref{fig:1}. As can be seen from Fig.~\ref{fig:1:b}, all algorithms for the gradient evaluation produce exactly the same result and give the same zero
point. Besides, the zero point coincides with the minimum point of the negative Log LF; see Fig.~\ref{fig:1:a}. Furthermore, it is readily seen that the obtained maximum likelihood estimate, $\hat \gamma^{MLE}_1$, coincides with the true value ${\gamma^*_1=2\cdot 10^{-4}}$. All these evidences substantiate our
theoretical derivations in Sections~\ref{sec:main},~\ref{sec:method}.

\subsection{Ill-Conditioned Test Problems}

Although we learnt from the previous examples that both methods, i.e. the conventional KF approach and the proposed array UD based scheme,
produce exactly the same results, numerically they no longer agree. To illustrate this, we consider the following ill-conditioned test problem.
\begin{exmp}
\label{ex:3:new}
Consider the state-space model~(\ref{eq2.1})--(\ref{eq2.2})
with $\{F, G, H, \Pi_0, Q, R \}$ given by
$$
F=I_3, G=0, R= I_2\delta^2\theta^2, \Pi_0=I_3 \theta^2, H=\left[
     \begin{IEEEeqnarraybox}[][c]{c/c/c}
      1 &  1 & 1\\
      1 &  1 & 1+\delta
      \end{IEEEeqnarraybox}
     \right]
$$ where $\theta$ is an unknown system parameter, that needs to be
estimated. To
simulate roundoff we assume that $\delta^2<\epsilon_{\text{roundoff}}$, but
$\delta>\epsilon_{\text{roundoff}}$ where $\epsilon_{\text{roundoff}}$ denotes the
unit roundoff error\footnote{Computer roundoff for floating-point
arithmetic is often characterized by a single parameter
$\epsilon_{\text{roundoff}}$, defined in different sources as the largest
number such that either $1+\epsilon_{\text{roundoff}} = 1$ or
$1+\epsilon_{\text{roundoff}}/2 = 1$ in machine precision.}.
\end{exmp}

When $\theta=1$, Example~\ref{ex:3:new} coincides with well-known
ill-conditioned test problem from~\cite{GrewalAndrews2001}. The difficulty to be explored is in the matrix
inversion. As can be seen, although ${rank \  H=2}$,  for any fixed value of the
parameter $\theta \ne 0$, the matrices ${R_{e,1}=R+H\Pi_{0}H^T}$ and
$(R_{e,1})'_{\theta}$ are ill-conditioned in machine precision, i.e. as $\delta \to
\epsilon_{\text{roundoff}}$. This leads to a failure of the
conventional KF approach and, as a result, destroys the entire
computational scheme. So, we plan to observe and compare performance
of the ``conventional KF technique'' and our new ``array UD based approach'' in such a situation. Additionally, we compare these techniques with the
``array SR based approach'' presented in~\cite{Kulikova2009IEEE}.  Also, we have to stress that Example~\ref{ex:3:new} represents numerical difficulties only for the covariance-type implementations. As a result, the SRIF based method~\cite{Bierman1990} can not be used  in our comparative study because of its information-type. This follows from~\cite{VerhaegenDooren1986} where numerical insights of various KF implementations are analyzed and discussed at large.

\begin{table*}
\renewcommand{\arraystretch}{1.3}
\caption{Effect of roundoff errors in ill-conditioned test problems} \label{MC-estimators}
\centering
{\scriptsize
\begin{tabular}{c|c||c|c|c||c|c|c||c|c|c}
\hline  & exact answer &
\multicolumn{3}{c||}{``differentiated" KF technique} &
\multicolumn{3}{c||}{array SR based approach} & \multicolumn{3}{c}{array UD based approach} \\
\hline $\delta$ & $\theta^*$ & Mean &  RMSE & MAPE &  Mean & RMSE
& MAPE &  Mean & RMSE
& MAPE \\
\hline
 $ 10^{-2\phantom{0}}$ & 7 &  6.9984   & 0.1243  &  1.4438 &  6.9984 & 0.1243  &  1.4438  & 6.9984 & 0.1243  & 1.4438\\
 $ 10^{-3\phantom{0}}$ & 7 &  7.0035   & 0.1233  &  1.4096 &   6.9996 & 0.1227  &  1.4011 & 7.0012 & 0.1227 & 1.4011 \\
 $ 10^{-4\phantom{0}}$ & 7 & 7.5116    & 1.1953  &  10.8291 & 7.0083 & 0.1111  & 1.2794  & 7.0083 & 0.1111 & 1.2794  \\
 $ 10^{-5\phantom{0}}$ & 7 & 5.4700   &  5.1658 & 72.0857 &  6.9966 & 0.1274  &  1.4706   & 6.9966 & 0.1274  & 1.4706\\
 $ 10^{-6\phantom{0}}$ & 7 &  4.0378   & 12.1748 & 157.6825 & 6.9979 & 0.1266  &  1.4581  & 6.9981 &  0.1264 & 1.4555\\
  \hline
\end{tabular}
}
\end{table*}

In order to judge the quality of
each above-mentioned computational technique, we conduct the following set of
numerical experiments. Given the ``true'' value of the parameter
$\theta$, say ${\theta^*=7}$, the system is simulated for $1000$
samples for various values of the problem conditioning parameter $\delta$. Then, we use the generated
data to solve the inverse problem, i.e. to compute the maximum
likelihood estimates by the above-mentioned  three different approaches. The same
gradient-based optimization method with the same initial value of
${\theta^{(0)}=1}$ is applied in all estimators.  More precisely, the optimization method utilizes the negative Log LF and its gradient
that are calculated by the examined techniques.  All methods are run in
MATLAB with the same precision (64-bit floating point).
We performed $250$ Monte Carlo simulations and report the posterior mean for
$\theta$, the root mean squared error (RMSE) and the
mean absolute percentage error (MAPE).

Having carefully analyzed the obtained results presented in Table~\ref{MC-estimators}, we conclude the following. When
${\delta=10^{-2}}$,  all the considered approaches produce exactly the same result. The posterior means from
all the estimators are all close to the ``true" value, which is equal to $7$. The RMSE and MAPE are equally small. Hence, we conclude that all three methods work similarly well in this well-conditioned situation. However, as ${\delta \to \epsilon_{\text{roundoff}}}$,
which corresponds to growing ill-conditioning, the conventional KF
approach degrades much faster than the array UD and SR based approaches. Already for ${\delta=10^{-4}}$, the RMSE
[MAPE] of the classical technique (i.e. the ``differentiated" KF) is
$\approx 11$ times [$\approx 8.5$ times] greater that of the  RMSE [MAPE] from the array UD and SR methods. For
${\delta\le 10^{-5}}$, the output generated by the ``differentiated"
KF is of no sense because of huge errors.
On the other hand, the new array UD based scheme and the previously proposed array SR method work robustly, i.e. with small errors, until
${\delta=10^{-6}}$, inclusively. Besides, we may note that both the array UD and SR approaches work equally well. Indeed, the RMSE and MAPE of
these techniques are approximately the same. Moreover, their results  change slightly for the above
$\delta$'s (see the last two panels in Table~\ref{MC-estimators}).

In conclusion, the ``differentiated" KF performs markedly worse
compared to the proposed array UD scheme (and previously known array SR approach). Moreover, both the UD and SR based techniques work equally well on the considered test problems. This creates a
strong argument for their practical applications. However, further investigation and comparative study have to be
 performed to provide  a rigorous theoretical and numerical
 analysis of the existing methods for the filter sensitivity computations.
These are open questions for a future research. Here, we just mention that the derived UD based method requires utilization of the MWGS orthogonal transformation while the array SR based approach uses any orthogonal rotation. Thus, the latter technique is more flexible in practice, but might produce less accurate answers for some problems because of its flexibility.

\bibliographystyle{IEEEtran}
\bibliography{BibTex_Library/books,%
              BibTex_Library/list_MVKulikova,%
              BibTex_Library/list_Tsyganova,%
              BibTex_Library/list_identification,%
              BibTex_Library/filters,%
              BibTex_Library/list_ML,%
              BibTex_Library/list_linalg}

\end{document}